\documentclass[%
 reprint,
 superscriptaddress,
nofootinbib,
amsmath,amssymb,
aps,longbibliography,
onecolumn,notitlepage,
]{revtex4-1}

\AtBeginDocument{%
    \newwrite\bibnotes
    \def\bibnotesext{Notes.bib}
    \immediate\openout\bibnotes=\jobname\bibnotesext
    \immediate\write\bibnotes{@CONTROL{REVTEX41Control}}
    \immediate\write\bibnotes{@CONTROL{%
    apsrev41Control,author="08",editor="1",pages="1",title="0",year="0"}}
     \if@filesw
     \immediate\write\@auxout{\string\citation{apsrev41Control}}%
    \fi
  }%

\usepackage{enumitem}
\usepackage{amsthm}

\newcommand{\tr}{\mathrm{tr}}
\newcommand{\innerp}[2]{\langle #1,#2\rangle}
\newcommand{\ot}{\otimes}
\newcommand{\cE}{\mathcal{E}}

\newcommand{\cS}{\mathcal{S}}
\newcommand{\cV}{\mathcal{V}}

\newcommand{\R}{\mathrm{R}}
\newcommand{\Loc}{\mathrm{L}}

\newcommand{\ins}{n_I}
\newcommand{\outs}{n_O}
\newcommand{\parties}{N}
\newcommand{\ns}{\mathrm{NS}}
\newcommand{\verts}{W}
\newcommand{\vertx}{e}
\newcommand{\edges}{M}
\newcommand{\edge}{m}

\newcounter{alg}

\newtheorem{theorem}{Theorem}

\newtheorem{lemma}[theorem]{Lemma}

\theoremstyle{definition}

\usepackage{hyperref}
\usepackage[nodayofweek]{datetime}
\usepackage{xcolor}
\definecolor{mylinkcolor}{rgb}{0,0,0.8}
\definecolor{citecol}{rgb}{0,0.6,0}
\hypersetup{unicode=true,%
   bookmarksnumbered=false,bookmarksopen=false,bookmarksopenlevel=1, %
   breaklinks=true,pdfborder={0 0 0},colorlinks=true}%
\hypersetup{%
   anchorcolor=mylinkcolor,citecolor=mylinkcolor, %
   filecolor=mylinkcolor,linkcolor=mylinkcolor, %
   menucolor=mylinkcolor,runcolor=mylinkcolor, %
   urlcolor=mylinkcolor}

\usepackage{graphicx}
\usepackage{bm}

\begin{document}

\title{Multi-system measurements in generalized probabilistic theories and their role in information processing}

\author{Giorgos Eftaxias} 
\email{giorgos.eftaxias@bristol.ac.uk}
\affiliation{Quantum Engineering Centre for Doctoral Training,
University of Bristol, Bristol BS8 1FD, United Kingdom}
\affiliation{Department of Mathematics, University of York, York YO10 5DD, UK}
\author{Mirjam Weilenmann}
\email{mirjam.weilenmann@unige.ch}
\affiliation{Institute for Quantum Optics and Quantum Information (IQOQI), Austrian Academy of Sciences, Boltzmanngasse 3, 1090 Vienna, Austria}
\affiliation{Département de Physique Appliquée, Université de Genève, Genève, Switzerland}
\author{Roger Colbeck}
\email{roger.colbeck@york.ac.uk}
\affiliation{Department of Mathematics, University of York, York YO10 5DD, UK}

\date{$12^{\text{th}}$ January 2024}

\begin{abstract}
Generalized probabilistic theories (GPTs) provide a framework in which a range of possible theories can be examined, including classical theory, quantum theory and those beyond.  In general, enlarging the state space of a GPT leads to fewer possible measurements because the additional states give stronger constraints on the set of effects, the constituents of measurements. This can have implications for information processing. In boxworld, for example, a GPT in which any no-signalling distribution can be realised, there is no analogue of a measurement in the Bell basis and hence the analogue of entanglement swapping is impossible. A comprehensive study of measurements on multiple systems in boxworld has been lacking. Here we consider such measurements in detail, distinguishing those that can be performed by interacting with individual systems sequentially (termed wirings), and the more interesting set of those that cannot.  We compute all the possible boxworld effects for cases with small numbers of inputs, outputs and parties, identifying those that are wirings. The large state space of boxworld leads to a small effect space and hence the effects of boxworld are widely applicable in GPTs. We also show some possible uses of non-wirings for information processing by studying state discrimination, nonlocality distillation and the boxworld analogue of nonlocality without entanglement.  Finally, we connect our results to the study of logically consistent classical processes and to the composition of contextuality scenarios. By enhancing understanding of measurements in boxworld, our results could be useful in studies of possible underlying principles on which quantum theory can be based.
\end{abstract}

\maketitle

\section{Introduction}
Standard textbook presentations of the postulates of quantum mechanics usually begin with a list of mathematical axioms with relatively little accompanying explanation. This is in contrast to relativity theory, for example, which can be based on the premise that the laws of physics are frame independent. Whether or how we can formulate quantum theory in a similar way remains open in spite of significant investigation (see, e.g.,~\cite{PR,Hardy,fuchs2002quantum,MM,Brukner2011,CDP}).
Quantum theory has some counter-intuitive features such as the presence of non-local correlations that seemingly defy classical explanation~\cite{Bell}. A general framework to study these features in the context of quantum theory and possible alternatives is that of generalized probabilistic theories (GPTs)~\cite{barrett}.  Beyond classical and quantum theory one well studied GPT is boxworld, which allows arbitrary no-signalling distributions to be realised. It is known for instance that particular cryptographic tasks remain possible even against adversaries that have access to boxworld systems~\cite{BHK,bcktwo,CR_free,GMTDAA}. In a further line of work that tries to single out quantum correlations within a range of alternative theories (see e.g.~\cite{van2013implausible,navascues2010glance,PPKSWZ,FSABCLA}), boxworld is a useful foil theory. In spite of these results, the structure of multi-system measurements in boxworld has not been developed in detail, but understanding these is important to fully characterize the information processing power of boxworld. Furthermore, although we refer to effects in boxworld throughout this paper, the measurements we find are applicable to a range of GPTs (see Section~\ref{sec:other_gpts}).

Further motivation for the study of multi-system measurements comes from the recent trend for studying information processing in quantum networks~\cite{Branciard2010,Fritz12,CSSH2020,PGT2022,SBCB2022}. It is likely that simple quantum networks will be built in the near future, hence it is useful to explore the possibilities these networks bring. GPTs provide a useful means to form a more general understanding of this. To compare what is possible in quantum theory as opposed to other GPTs it is furthermore necessary to understand the structure of multi-system measurements in the latter. Such comparison furthermore informs the question of in what sense quantum theory is optimal for information processing~\cite{WC4}. In the present work, we thus explore such multi-system measurements.  In addition, we connect our results to the study of logically consistent classical processes~\cite{BaumelerWolf} and of composition of contextuality scenarios~\cite{Acin_Contextuality,Wolfe_Sainz}, pointing to a new kind of composition in the latter case.

Our study of the set of possible measurements in boxworld proceeds using the set of effects, the constituent parts of measurements (see later).  The large set of possible states in boxworld comes at the expense of having a smaller set of effects than in theories with weaker correlations. In this work we will be interested in the set of possible measurements that can be performed with access to several systems in boxworld. One type of such measurements are the \emph{wirings}~\cite{barrett}, which correspond to processes in which a measurement is applied to one system, then to a second depending on the result and so on (or convex combinations of such processes). Wirings are hence implementable using local operations and classical communication (LOCC) when the individual systems are separated\footnote{More precisely, wirings only require one-way LOCC.}. It is well-known that in quantum theory not all measurements are of this form. For example, a measurement in the Bell basis as used in teleportation cannot be realised by wiring together measurements on individual systems. For single and bipartite systems in boxworld, all measurements are wirings~\cite{barrett}, while this is no longer true for three or more systems~\cite{ShortBarrett}. One of the aims of the present paper is to find the multi-system measurements that are not wirings and investigate their significance for information processing. 

We classify scenarios by the number of systems, and the number of inputs and outputs for the boxes of each system. One way to find the set of all extremal effects in a given scenario is using vertex enumeration (all the extremal no-signalling distributions act as facets since the inner product of each such state with any valid effect must be greater than $0$ and less than $1$). However, directly performing vertex enumeration is slow, except in the smallest scenario~\cite{SPG}. To circumvent this we exploit an alternative method for finding extremal effects that starts with different ways to represent the identity effect.  
By breaking down these identity effects we can find the extremal effects in various scenarios. We separately consider deterministic and non-deterministic effects and classify the effects into wirings and non-wirings before investigating the significance of the latter for state discrimination and nonlocality distillation, showing the advantages of non-wiring measurements.  We also discuss how our findings relate to the phenomenon of quantum nonlocality without entanglement~\cite{BDFMRSSW}, which is the existence of measurements comprising product effects that cannot be performed by LOCC. The non-wiring type measurements of the present paper serve as a boxworld analogue of these, and we find a set of product states in boxworld that can be perfectly distinguished with a non-wiring measurement, but not with any wiring.

Understanding measurements in GPTs is also useful to gain insight into which features of quantum theory make it special with respect to other theories, which in turn may help find underlying principles on which quantum theory can be based. Boxworld is known to have only separable measurements~\cite{ShortBarrett} (i.e., those for which all effects  can be expressed as a sum of product effects), which in a sense makes it inferior to other GPTs. However, this property also makes boxworld a suitable example for studying the difference between separable measurements that can be done using LOCC and those that cannot, and of the information processing capabilities enabled by the latter, an understanding that will likely carry over to arbitrary GPTs.  Further insights into this difference are also desirable in the quantum case~\cite{MancinskaLocc,CheflesLocc,ChitambarMonotones,childsFramework}.

The rest of the paper proceeds as follows. We first give a short introduction to GPTs, before introducing the notation and technical background in Section~\ref{sec:pre}. In Section~\ref{sec:extremal}, we present technical results underlying the algorithms for generating all extremal effects in box-world; these algorithms are presented in Section~\ref{sec:algorithm}. In Section~\ref{sec:examples} we apply our algorithm to various scenarios of up to 4-inputs and up to 4-outputs and demonstrate the application of these results in information processing tasks in Section~\ref{sec:application}. Finally, we connect our work to logically consistent classical processes and to composition of contextuality scenarios in Section~\ref{sec:beyond_GPT} before making a few concluding remarks in Section~\ref{sec:conclusion}.

\section{Background}
We briefly outline the framework for GPTs that will be used throughout this paper.  In a GPT, states are represented as vectors in a vector space $V$.  We use $\cS\subseteq V$ to denote the set of all possible states (note that the state space depends on the size of the system), and we typically assume $\mathcal{S}$ is convex and compact. An \emph{effect} is a linear map from a state to a probability. Effects can be taken as vectors in the vector space dual to $V$, which we call $\cE$, and the map from states to probabilities is then formed by taking the inner product between the state and the effect. Given a state space, a valid effect $e$ must satisfy $0\leq\innerp{e}{s}\leq1$ for all $s\in\cS$. Under the \emph{no restriction hypothesis}~\cite{chiribella2010probabilistic}, which we will assume here, this necessary condition is taken to be necessary and sufficient for an effect to be valid.

The state space of a composite system is formed by taking some kind of tensor product between the individual state spaces. A minimal requirement is that if $s_A\in\cS_A$ is a state of system $A$ and $s_B\in\cS_B$ is a state of system $B$, then $s_A\ot s_B$ is a state of $AB$. If all joint states have this form, or are convex combinations of states of this form, then the joint state space is said to be the minimal tensor product of the individual state spaces.  Taking the composite system effect space to obey the no restriction hypothesis, the min tensor product state space corresponds to forming the effect space by taking the maximal tensor product of the individual effect spaces.

One way to specify a state is via the probabilities of the outcomes of a set of \emph{fiducial measurements}~\cite{Hardy}, i.e., a set that is sufficient to completely characterize the state.  We will work with systems that have a finite set of fiducial measurements, each of which has a finite number of outcomes.  This means that a state can be specified using a vector whose entries contain the probabilities of outcomes for the fiducial measurements. For simplicity, we consider cases in which the number of outcomes of each fiducial measurement is the same.  For instance, if for a single system there are two 2-outcome fiducial measurements, we write the state as
\begin{equation}\left(P(0|0),P(1|0)\mid P(0|1),P(1|1)\right), \label{eq:notation}\end{equation}
where $P(i|j)$ is the probability of output $i$ given input $j$.

Although this is a four dimensional vector, due to normalization there are only two independent parameters. For later convenience we stick with the larger representation rather than suppressing the redundant parameters.

In this work we assume \emph{local tomography}~\cite{hardy_2011}, i.e., that the state of a joint system can be determined from the statistics of local fiducial measurements (called the global state assumption in~\cite{barrett}).  This means that a state of $\parties$ parties (where for each single system there are two 2-outcome fiducial measurements) has a similar representation:
\begin{align}\label{eq:state_rep}
  \left(P(0\ldots0|0\ldots0),\ldots,P(1\ldots1|0\ldots0)\mid\ldots\mid P(0\ldots0|1\ldots1),\ldots,P(1\ldots1|1\ldots1)\right)
\end{align}
(the ordering is such that first for input $00\ldots0$ the outcomes increase counting in binary, and then the inputs increase counting in binary).

For every system there is an identity effect, i.e., an effect $u$ such that $\innerp{u}{s}=1$ for all $s\in\cS$. In our notation there are several ways to write this effect.  For a single system we can write $u^{\R}=(1,1|0,0)$, $u^{\R}=(0,0|1,1)$, $u^{\R}=(1,1|1,1)/2$ etc. Although written differently, these all represent the same effect, the notation $u^{\R}$ meaning a representation of the effect $u$. In general, we identify two vectors as representing the same effect if they have the same inner product with all elements of the state space.

The set of vectors that can be added to any effect vector without changing the effect it represents we term \emph{no-signalling moves} because they each arise as a result of the state space $\cS$ only containing no-signalling distributions (or as a result of the normalization). For instance, for all valid states, $P(00|00)+P(10|00)-P(00|10)-P(10|10)=0$, which represents the impossibility of Alice's choice of measurement affecting Bob's outcome when Bob makes input $0$.  This no-signalling condition can be encoded using a vector $r$ such that $\innerp{r}{s}=0$ for all $s\in\cS$. 
Thus, if $e^{\R}$ is a vector representing a valid effect, then $e^{\R}+r$ represents the same effect. We use $\{r_i\}$ to denote a complete set of generators of all such vectors, so that any no-signalling move can be written as a linear combination of vectors from $\{r_i\}$.

A \emph{measurement} is a collection of effects that sum to the identity effect, i.e., given state space $\cS$, the set of effects $\{e_1,\ldots,e_m\}$ form a measurement on $\cS$ if for each $i=1,\ldots,m$ we have $0\leq\innerp{e_i}{s}\leq1$ for all $s\in\cS$ and if $\sum_{i=1}^me_i=u$.

In this work we will be interested in the effect space when the state space comprises all no-signalling distributions (also known as \emph{boxworld}). Because boxworld has large multipartite state spaces, the set of possible effects is comparably restricted. Boxworld effects are therefore valid in a wide range of GPTs. We consider the effect spaces for various small numbers of inputs, outputs and parties (larger cases become too computationally intensive).  Due to linearity, the effect space is convex, and can hence be characterized in terms of a set of extremal effects. There are a finite number of these, generating a convex polytope; our aim is to have a procedure that can generate the vertices of this polytope.

In the two-party case, the complete set of extremal effects was computed in~\cite{SPG}. It was found that there are 82 such effects. These were computed by vertex enumeration starting from the facet description of the effect space.  This facet description says that to be a valid effect $0\leq\innerp{e}{s_i}\leq1$ for all extremal states $s_i$ (in the present case there are 24 extremal states: 16 local deterministic states and 8 PR-box-type states). In this work we use a different method that allows us to treat cases for which our computational tools for vertex enumeration are prohibitively slow.

We will also be interested in a special type of effect called a \emph{wiring}. Wirings are effects that can be implemented by convex combinations of procedures of the following form: choose one of the states and choose a measurement to make on that state, take its output and apply a function to it to choose the next state and the measurement on the next state, and so on, where at each step all the previous outputs are used as arguments of the function that selects the next state and the measurement performed on it.  The final output is then formed by taking a function of all the individual measurement outcomes --- see Fig.~\ref{fig_wir_effect_example} for an illustration.  
\begin{figure}[h]
    \centering
    \includegraphics[scale=0.55]{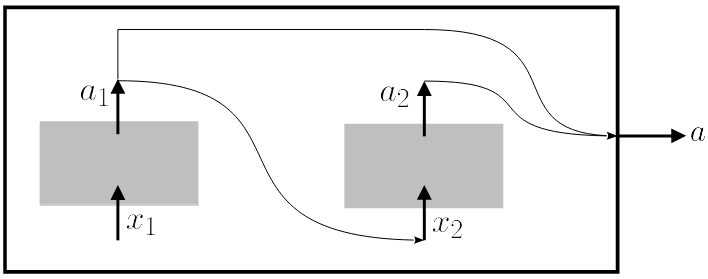}
    \caption{Illustration of a wiring-effect in the bipartite case. An input $x_1$ is made generating outcome $a_1$, then $x_2$ is taken as a function of $a_1$ giving outcome $a_2$. The final output, $a$, is then a function of $a_1$ and $a_2$. For instance, if $x_1=0$, $x_2=a_1$ and $a=a_1\oplus a_2$ then the effect corresponding to $a=0$ is ${e_{a=0}=(1,0,0,0|0,0,0,1|0,0,0,0|0,0,0,0)}$ and the effect corresponding to $a=1$ is ${e_{a=1}=(0,1,0,0|0,0,1,0|0,0,0,0|0,0,0,0)}$.} 
    \label{fig_wir_effect_example}
\end{figure}
The extremal wirings are effects of this type that cannot be represented as convex combinations of others. These boxworld measurements are of particular interest, since they correspond to classical processing of inputs and outcomes and are thus applicable in any operational theory allowing classical processing. In particular, they are applicable irrespective of the type of system under consideration and can be used for instance for non-locality distillation, which in turn may have applications for device-independent information processing. This is discussed further in Section~\ref{sec:NLdist}.

\section{Preliminaries}\label{sec:pre}
Let $\ins$ be the number of inputs and $\outs$ be the number of outputs per party with $\parties$ parties. A state can then be represented as a $(\ins\outs)^\parties$-dimensional vector, whose entries are the conditional probabilities of every set of outputs given every set of inputs. Given a subset, $V$, of the set $[\parties]:=\{0,1,2,\ldots,\parties-1\}$ of parties, we use $X_V$ as the random variable representing their inputs (elements of $[\ins]^{|V|}$) and $A_V$ as that for their outputs (elements of $[\outs]^{|V|}$). 

We will now define our state space $\cS_{\ns}$. To be elements of this we require three conditions to hold: no-signalling, positivity and normalization, which we detail below. The no-signalling condition is that no subset of parties should be able to signal to any other subset of parties. In other words, for any disjoint subsets $V$ and $W$ of $[\parties]$, the parties in $V$ cannot signal to those in $W$. Mathematically, this means
\begin{align}\label{eq:NS}
&\sum_{{\bf a}_V}P_{A_WA_V|X_W={\bf x}_W,X_V={\bf 0}}({\bf a}_W,{\bf a}_V)=\sum_{{\bf a}_V}P_{A_WA_V|X_W={\bf x}_W,X_V={\bf x}'_V}({\bf a}_W,{\bf a}_V)\\ \text{for all}\ \  &{\bf a}_W\in\{0,1,2,\ldots,\outs-1\}^{|W|},\ {\bf x}_W\in\{0,1,2,\ldots,\ins-1\}^{|W|}\ \text{and}\ {\bf x}'_V\in\{0,1,2,\ldots,\ins-1\}^{|V|}\,.\nonumber
\end{align}
A smaller subset of this set of conditions is sufficient to generate the whole set, namely it suffices that no party can signal to the collection of all the others. This is stated formally in the following lemma.
\begin{lemma}\label{lem:NSmoves}
  The set of conditions implied by~\eqref{eq:NS} for all disjoint subsets $V$ and $W$ of $[\parties]$ is implied by demanding that~\eqref{eq:NS} holds for the subset of conditions where $V$ is a singleton and $W=[\parties]\setminus V$.
\end{lemma}

The conditions where $V$ is a singleton and $W=[\parties]\setminus V$ can be represented by a set of vectors $\{r_i\}$, such that $\innerp{r_i}{s^{\R}}=0$ for all $s\in\cS_{\ns}$, and where each $r_i$ contains $\outs$ elements with value $1$, $\outs$ elements with value $-1$ and the remaining entries are $0$.  The number of such conditions is $\parties(\outs\ins)^{\parties-1}(\ins-1)$. These are not all linearly independent: the dimension of the space of un-normalized no-signalling distributions is $(\ins(\outs-1)+1)^\parties$~\cite{Cirelson93,RBG}.

A valid state must also have positive entries, since each entry represents a probability, i.e.,
\begin{equation}\label{eq:pos}
P_{A_S|X_S}({\bf a}|{\bf x})\geq0\ \ \text{for all}\ \ {\bf a},{\bf x}\,,
\end{equation}
where $S=[\parties]$.

The remaining condition for a vector to be an element of $\cS_{\ns}$ is that
\begin{equation}\label{eq:norm}
\sum_{a_S}P_{A_S|X_S={\bf 0}}({\bf a}_S)=1\,,
\end{equation}
where $S=[\parties]$.  Note that, given the no-signalling conditions, if~\eqref{eq:norm} holds for $X_S={\bf 0}:=0\ldots0$ then it also holds for $X_S={\bf x}_S$ for any other ${\bf x}_S\in[\outs]^{|S|}$. Furthermore, together with~\eqref{eq:pos} this implies that no entry can exceed $1$ (so conditions that each probability is at most $1$ do not need to be added).

An important set of states are the \emph{local deterministic states}.  These correspond to deterministically assigning one of the outcomes for every possible input of every party. Thus, they are states of the form
$$P_{A|X={\bf x}}({\bf a})=P_{A_1|X_1=x_1}(a_1)P_{A_2|X_2=x_2}(a_2)\ldots P_{A_\parties|X_\parties=x_\parties}(a_\parties)\,,$$
and where $P_{A_i|X_i}\in\{0,1\}$ in all cases.  We denote these $s^{\Loc,i}$, where $i$ runs from $1$ to $(\ins\outs)^\parties$. The linear span of the set of local deterministic states is the same as that of the no-signalling distributions (being local does not confer any additional \emph{equality} constraints over being no-signalling). Thus, given two vectors representing effects, we can check whether they represent the same effect by checking whether both vectors have the same inner product with all local deterministic distributions\footnote{In general, a subset of these would also suffice.}.

An important symmetry class for a given state or effect is that of the relabellings.  We can relabel the parties, the inputs for each party, or the outputs for each input of each party.  Thus, the total size of the symmetry group of relabellings is $((\outs!)^{\ins}\ins!)^\parties\parties!$.

\section{Extremal effects} \label{sec:extremal}
We are interested in finding the extremal effects for a given scenario, i.e., for a given $(\parties,\ins,\outs)$. As mentioned before we use a representation that has some redundancy in that the constraints arising from the no-signalling moves and normalisation are not used to reduce the parameters. (Such a reduction would typically be done when trying to solve this problem directly by means of vertex enumeration.) Instead, the technical results of this section show that the redundant representation employed has a convenient structure that allows us to find extremal effects in a different way, which will be the basis of the algorithms we present in Section~\ref{sec:algorithm}.
We start with the following observation, expressed in terms of a \emph{standard basis vector} (SBV), by which we mean a vector with all components 0 except for one 1.
\begin{lemma}\label{lem:extremal1}
Every $(\ins\outs)^\parties$-dimensional SBV is a representation of an extremal effect.
\end{lemma}
\begin{proof}
  To see that SBVs all represent effects, note that the inner product of such a vector $e^{\R}$ with a state $s\in\cS$ gives a single probability, hence $0\leq\innerp{e^{\R}}{s^{\R}}\leq1$ for all $s\in\cS$, making $e^{\R}$ a representation of a valid effect.  By considering a state satisfying $\innerp{e}{s}=1$, it is clear that $\alpha e$ is not an effect for any $\alpha>1$, hence $e$ is on the boundary of the effect space. To see that it is extremal, suppose $e=\alpha e_1+(1-\alpha)e_2$ for two effects $e_1$ and $e_2$ and $0<\alpha<1$. Then
  $$\innerp{e}{s}=\alpha\innerp{e_1}{s}+(1-\alpha)\innerp{e_2}{s}\,.$$
  Now suppose $s_1$ is a state with $\innerp{e}{s_1}=1$. Since $e_1$ and $e_2$ are effects, it follows that $\innerp{e_1}{s_1}=\innerp{e_2}{s_1}=1$. Similarly, if $s_0$ is a state with $\innerp{e}{s_0}=0$, it follows that $\innerp{e_1}{s_0}=\innerp{e_2}{s_0}=0$. Because the local deterministic distributions are $\{0,1\}$-valued, if $s$ is local deterministic it satisfies $\innerp{e}{s}\in\{0,1\}$. It follows that both $e_1$ and $e_2$ have the same action on any local deterministic distribution as $e$ does.  Since the set of local deterministic distributions span the space of no-signalling distributions, $e_1$ and $e_2$ must have the same action as $e$ for any $s\in\cS_\ns$. Thus, $e_1$, $e_2$ and $e$ must be the same effect and $e$ is extremal.
\end{proof}

The $\{0,1\}$-valued representations of the identity effect will be important because these are the class that can be used to generate all wirings (as well as some non-wirings).

\begin{lemma}\label{lem:eff}
Let $u^{\R}$ be a $\{0,1\}$-valued vector representing the identity effect. Any vector $e^{\R}$ formed from $u^{\R}$ by replacing any number of the $1$ entries with $0$s is a representation of an extremal effect.
\end{lemma}
\begin{proof}
  First, since $u^{\R}$ represents the identity effect we have $\innerp{u^{\R}}{s^{\R}}=1$ for all $s\in\cS$. Since $u^{\R}$ is a $\{0,1\}$-valued vector we can write it in terms of the SBV effects $\{e_i\}$ as $u=\sum_i\lambda_ie_i$, where $\lambda_i\in\{0,1\}$. We can also write $e^{\R}=\sum_i\lambda'_ie^{\R}_i$, where $\lambda'_i\in\{0,1\}$ and $\{i:\lambda'_i=1\}\subset\{i:\lambda_i=1\}$.  It follows that $0\leq\innerp{e}{s}\leq\innerp{u}{s}=1$ for all $s\in\cS$, so $e$ is a valid effect.

  To see that it is extremal, consider writing $e=\alpha e_1+(1-\alpha)e_2$ for two effects $e_1$ and $e_2$ and $0<\alpha<1$. Since $e^{\R}$ is a $\{0,1\}$-valued representation of $e$, and all local deterministic states have a $\{0,1\}$-valued representation, for a local deterministic state $s$ we have $\innerp{e}{s}\in\{0,1\}$. Hence, by the same argument as in Lemma~\ref{lem:extremal1}, $e$ must be extremal.
\end{proof}

\begin{lemma}\label{lem:01val}
Any extremal wiring has a representation that can be formed by taking the $\{0,1\}$-valued representation of the identity effect corresponding to all parties making the input $0$, applying no-signalling moves in such a way that it remains $\{0,1\}$-valued, and replacing some of the entries that are $1$ with $0$.
\end{lemma}
\begin{proof}
  Consider first the case $\parties=1$. In this case the extremal wirings are the effects formed by choosing one of the $\ins$ measurements and then applying a function from $\{0,1,\ldots,\outs-1\}$ to $\{0,1,\ldots,\outs-1\}$ to the outcome. In this case the no-signalling moves\footnote{The term no-signalling doesn't make sense for $\parties=1$, but the mathematical conditions are well-defined and take the form $\sum_aP_{A|X=x}(a)=\sum_aP_{A|X=x'}(a)$ for all $x,x'$.} take us from the identity for input $0$ to that for all other choices of measurement so the statement holds.

  Now assume by induction that the statement holds for $\parties-1$ parties and consider $\parties$ parties.  The first step in an extremal wiring is to choose one of the $\parties$ boxes and make a fixed input to that box.  Up to symmetry, we can assume the first box is chosen and input $0$ is made (relabelling parties or inputs does not change whether an effect is a wiring or not).  In this case, the effect can only have non-zero entries where these correspond to elements of the state of the form $P(a_1a_2\ldots a_{\ins}|0x_2\ldots x_{\ins})$, i.e., where $x_1=0$.  For each outcome we can then consider the $\parties-1$ party effect that is performed conditioned on $x_1=0$ and the value of $a_1$.  By assumption, each of these sub-effects can be formed by taking the $\{0,1\}$-valued representation of the identity effect corresponding to $\parties-1$ parties making the input $0$, applying no-signalling moves in such a way that it remains $\{0,1\}$-valued, and replacing some of the entries that are $1$ with $0$.  Let us ignore the replacement of $1$ entries with $0$s for the moment and consider only the representation of identity.  Up to no signalling moves each of the sub effects corresponds to all parties measuring $0$, and, up to symmetry the measurement on the first box corresponds to $x_1=0$. Thus, up to no-signalling moves, the identity is that corresponding to all parties measuring $0$, and the effect is then formed by zeroing entries of this.
\end{proof}

\begin{lemma}\label{lem:pos}
Every effect can be represented by a $(\ins\outs)^\parties$-dimensional vector in which every entry is non-negative.
\end{lemma}
\begin{proof}
  This is proven as part of Theorem~7 in~\cite{barrett}. We give the argument for completeness.  Consider the cone of non-normalized states in our representation, i.e., the set of $(\ins\outs)^\parties$-dimensional vectors $$\cV=\{v:\innerp{e_i}{v}\geq0 \ \forall i,\, \innerp{r_j}{v}=0\ \forall j\},$$
where $\{e_i\}$ are the SBV effects and $r_j$ are vectors representing the no-signalling moves. We can rewrite $\innerp{r_j}{v}=0$ as $\innerp{r_j}{v}\geq0$ and $\innerp{-r_j}{v}\geq0$. The dual cone is then that formed by the conic hull of $\{e_i\}_i\cup\{r_j\}_j\cup\{-r_j\}_j$. Thus, any effect can be written as $\sum_it_ie_i+\sum_jw_jr_j$, where $t_i\geq0$, but $w_j$ can be negative.  Since $\innerp{r_j}{v}=0$ for all $v\in\cV$, one representation of the effect is when the values of $\{w_j\}$ are set to zero. Thus, any effect can be written in the form $\sum_it_ie_i$ where $t_i\geq0$.
\end{proof}

It is helpful to consider the set of $(\ins\outs)^\parties$-dimensional identity effects that have positive entries.  These form a convex polytope, since they are defined by the vectors $u^{\R}$ for which every element is positive and such that $\innerp{s^{\Loc,i}}{u^{\R}}=1$ where $i$ runs over all local deterministic distributions. There are hence $(\ins\outs)^\parties$ equality constraints and $(\ins\outs)^\parties$ inequality constraints. We call the extreme points of this polytope the \emph{extremal representations of the identity effect}, and these can be computed using vertex enumeration.

\begin{lemma}\label{lem:areall}
Every extremal effect has a representation as a vector that can be formed by taking an 
extremal representation of the identity effect and replacing some of the non-zero entries with $0$s.
\end{lemma}
\begin{proof}
  Let $e$ be an effect and $f=u-e$.  By Lemma~\ref{lem:pos}, we can represent $e$ and $f$ using $(\ins\outs)^\parties$-dimensional vectors $e^{\R}$ and $f^{\R}$ whose entries are non-negative. Write $e^{\R}=\sum_i\lambda_ie_i$ and $f^{\R}=\sum_i\mu_ie_i$, where $\{e_i\}$ are the SBV effects and $\{\lambda_i\}$ and $\{\mu_i\}$ are non negative. Thus, $u^{\R}=\sum_i(\lambda_i+\mu_i)e_i$ is a representation of the identity effect.

  We claim that if $e$ is extremal then for each $i$ either $\lambda_i=0$ or $\mu_i=0$.  Suppose by contradiction that there is some $j$ for which $0<\lambda_j<\lambda_j+\mu_j$. Then
  \begin{align*}
    e^{\R}&=\lambda_je_j+\sum_{i\neq j}\lambda_ie_i=\frac{\lambda_j}{\lambda_j+\mu_j}\left((\lambda_j+\mu_j)e_j+\sum_{i\neq j}\lambda_ie_i\right)+\left(1-\frac{\lambda_j}{\lambda_j+\mu_j}\right)\sum_{i\neq j}\lambda_ie_i\\
    &=\frac{\lambda_j}{\lambda_j+\mu_j}(e^{\R}+\mu_je_j)+\left(1-\frac{\lambda_j}{\lambda_j+\mu_j}\right)(e^{\R}-\lambda_je_j)\,.
  \end{align*}
Both $e^{\R}+\mu_je_j$ and $e^{\R}-\lambda_je_j$ represent effects (for the former, note that $e^{\R}+\mu_je_j=u^{\R}-(f^{\R}-\mu_je_j)$). Hence we have decomposed $e$ as a convex combination of other effects, contradicting the assumption that $e$ is extremal.  Note that this implies that for any extremal effect, there exists a representation of it and its complement that are orthogonal\footnote{In fact, what we have shown is even stronger: there is a representation in which the element-wise product of $e^{\R}$ and $f^{\R}$ is zero.}. Hence, we have shown that if $e$ is extremal, it can be formed by zeroing entries from a representation of an identity effect.  It remains to show that non-extremal representations of the identity effect need not be used.

Suppose $u^{\R}$ is a non-extremal representation of the identity effect, so $u^{\R}=\sum_i \nu_i u^{\R}_i$ with $\{u^{\R}_i\}$ being extremal representations of the identity effect, $\nu_i \geq 0$, $\sum_i \nu_i=1$ and at least two $\nu_i>0$. Let $Z$ be a map that zeroes some of the entries and suppose $e^{\R}=Z(u^{\R})=\sum_i \nu_i Z(u^{\R}_i)$. If there are two or more values of $i$ for which both $\nu_i>0$ and $Z(u^{\R}_i)\neq0$, then $e^{\R}$ is not extremal. If there is only one $i$ such that both $\nu_i>0$ and $Z(u^{\R}_i)\neq0$, then $e^{\R}$ is also not extremal (but is proportional to the zeroing of an extremal representation of the identity effect).
\end{proof}

\section{Computing extremal effects} \label{sec:algorithm}
Our method for computing extremal effects is suggested by Lemma~\ref{lem:areall}. The first step is to identify all extremal representations of the identity effect by means of a vertex enumeration.

Specifically, the condition $\innerp{u}{s}=1$ for all elements of the state space can be imposed by requiring 
$\innerp{u}{s}=1$ for all local deterministic $s$, since the linear span of the local deterministic states covers the state space. This means that the set of identity effects for $\parties$ parties can be expressed using $4^\parties$ equality constraints and $4^\parties$ inequality constraints (positivity of the individual entries). This problem is more tractable than performing the full vertex enumeration to compute all extremal effects directly. Instead, the extremal effects are obtained from the identity effects by considering sub-effects (see Lemma~\ref{lem:areall}).

Computing the set of all extremal effects scales badly with the parameters of the scenario and hence we do not compute all in most cases (although our algorithm would in principle allow this). In several scenarios we instead compute the deterministic extremal effects, which include all extremal wirings.

\subsection{Computing deterministic extremal effects}
Our method for computing all deterministic extremal effects is suggested by Lemma~\ref{lem:eff}. This algorithm recovers all deterministic extremal effects, which includes all wirings (cf.\ Lemma~\ref{lem:01val}).

We first find all the $\{0,1\}$-valued representations of the identity effect.  This can be done by the following algorithm:

\phantom{1}

\refstepcounter{alg}
\noindent{\bf Algorithm \arabic{alg} -- generate $\{0,1\}$-valued identity effects}\label{alg:id}
\begin{enumerate}
  \item Let $S=\{u^{\R}_1\}$ where $u^{\R}_1$ is any $\{0,1\}$-valued representation of the identity effect. [The algorithm can also be started with any initial set of $\{0,1\}$-valued representations of the identity effect.]
  \item\label{st:2} Generate $S'=\{s_j\pm r_i\}_{i,j}$, where $s_j$ are elements of $S$ and $r_i$ are a complete set of generators of the no-signalling moves taking values $\{-1,0,1\}$ (cf.\ Lemma~\ref{lem:NSmoves}).
  \item\label{st:3} Remove elements of $S'$ with negative entries and set $S=S'$.
  \item\label{st:4} Repeat steps~\ref{st:2} and~\ref{st:3} until $S$ stops increasing.
  \item Output $S$.
\end{enumerate}

Algorithm~\ref{alg:id} is a sub-algorithm of our main routine:

\phantom{1}

\refstepcounter{alg}
\noindent{\bf Algorithm \arabic{alg} -- generate all $\{0,1\}$-valued effects}\label{alg:all}
\begin{enumerate}
  \item Use Algorithm~\ref{alg:id} to generate all representations of the identity effect.
  \item\label{st:12} For each representation of the identity, form a new set of effects containing all effects that are obtained by deleting any number of $1$s from each of the representations.
  \item Take the union of all the sets generated.
  \item\label{st:14} For each element $e_i$ in this union, compute $\left(e_i,L(e_i)\right)$, where $L(e_i)=\left(\innerp{e_i}{s^{\Loc,1}},\innerp{e_i}{s^{\Loc,2}},\ldots\right)$ and generate a list $S=\left((e_1,L(e_1)),(e_2,L(e_2)),\ldots\right)$ of all these pairs.
  \item\label{st:15} Go through the list checking whether $L(e_i)=L(e_j)$ where $i\neq j$. If so, remove one of the two elements from $S$.
    \item Output $S$.
\end{enumerate}

Note that each identity effect has $\outs^\parties$ elements with value $1$, so there are $2^{\outs^\parties}$ ways of deleting $1$s for each representation of the identity effect in Step~\ref{st:12} (hence this algorithm does not scale well as the number of parties increases). Taking the union of the sets involves removing any duplicate representations.  However, at this point there remain different representations of the same effect in our set. The use of $L(e_i)$ is a convenient way to remove such different representations.

We are also interested in classifying the extremal effects as either wirings or non-wirings.  We first classify a set of \emph{wiring representations} using the following algorithm (cf.\ Lemma~\ref{lem:01val}):

\phantom{1}

\refstepcounter{alg}
\noindent{\bf Algorithm \arabic{alg} -- classify a representation of an effect as wiring representation}\label{alg:wir}

This algorithm takes as input a representation $e^\R$ of an effect.
\begin{enumerate}
\item\label{st:31} If the number of parties is 1, return $1$.\footnote{See later for an explanation as to why we can increase the number of parties to 2 in this step.}
\item Run over all exchanges of party 1 with each other party, and exchanges of input labels for the chosen party until the resulting effect (after the exchanges) has zero elements wherever $x_1\neq0$. If no such case is found return $0$.
\item For each of the possible outcomes $a_1$ compute the $\parties-1$ party effect conditioned on $x_1=0$ and the value of $a_1$ (there are $\outs$ instances to compute).
  \item Recursively run the same algorithm on each of these $\parties-1$ party effects. If all cases return $1$ then return $1$, otherwise return $0$.
\end{enumerate}

Algorithm~\ref{alg:wir} outputs $0$ if its input is not a wiring representation and outputs $1$ if it is. An effect is a wiring if and only if it can be expressed as a convex combination of effects that have wiring representations. An extremal effect is hence a wiring if and only if it has a wiring representation.

To connect Algorithm~\ref{alg:wir} to the previously mentioned notion of a wiring, consider the first step in an extremal wiring. This involves choosing one of the boxes to make an input to as well as the value of the input. Consider the case where this is the first box and the input made is $x_1=0$. In this case, $e^\R$ will have zeros for the elements corresponding to any probabilities conditioned on other values of $x_1$. Hence, if $e^\R$ is a wiring, there must exist a permutation $\Pi$ of parties and of input labels such that $\Pi e^\R$ has zero elements wherever $x_1\neq0$.  For each of the possible outcomes $a_1$ we can consider the $\parties-1$ party effect conditioned on $x_1=0$ and the value of $a_1$ (that each of these are an effect follows by considering the set of states that are a tensor product of the deterministic state that always outputs $a_1$ for $x_1=0$ with any $\parties-1$ party state on the remaining systems). We can check that all these smaller effects have an analogous property in the same way and recurse.  If the required label permutation exists at all levels, we can conclude that the effect corresponds to a wiring.

Suppose we run Algorithm~\ref{alg:wir} on a particular $\{0,1\}$-valued representation $e^\R$ of an effect $e$. If we get output $1$ then we know that $e$ can be implemented as a wiring. However, if we get $0$, it could be that there is an alternative representation of $e$ that is a wiring representation. To understand which effects are wirings or not we can modify Steps~\ref{st:14} and \ref{st:15} of Algorithm~\ref{alg:all} to the following:
\begin{enumerate}[label=\arabic{enumi}$'$]\setcounter{enumi}{3}
\item\label{st4p} For each element $e_i$ in this union, compute $\left(e_i,L(e_i),W(e_i)\right)$, where $L(e_i)$ is as before and $W(e_i)$ is whether the representation is a wiring representation or not and generate a list $S=\left((e_1,L(e_1),W(e_1)),(e_2,L(e_2),W(e_2)),\ldots\right)$ of all these pairs.
\item\label{st5p} Go through the list checking whether $L(e_i)=L(e_j)$ where $i\neq j$. If so, check whether $W(e_i)=1$. If so, remove element $j$ from the list, otherwise remove element $i$.
\end{enumerate}

With this modified algorithm, we can identify effects as either wirings or non-wirings as they are found.

\section{Results} \label{sec:examples}
\subsection{Two inputs and two outputs per party}
In this section we discuss cases with $\ins=2$ and $\outs=2$ for $\parties=2,3,4$.

\subsubsection{$\parties=2$}
This case was already computed in~\cite{SPG} using a different procedure. The output of our algorithm in this case agrees with that of~\cite{SPG}. In particular, there are $82$ extremal effects. We break these down into $7$ classes (two effects are in the same class if they are equivalent up to relabellings, where we can relabel the parties, the inputs for each party and the outputs for each input of each party)\footnote{Note that in~\cite{SPG} they classify differently into $5$ classes.}. All of the effects are wirings, which was already known from~\cite{ShortBarrett}.  That all of the effects are wirings means we can alter Step~\ref*{st:31} of Algorithm~\ref{alg:wir} to ``if the number of parties is 2, return $1$''.

In this case the deterministic effects are sufficient for describing the full effect polytope, so the more general case is omitted here.

\subsubsection{$\parties=3$}
A complete description of all deterministic effects was not known for this case, although in~\cite{ShortBarrett} it was noted that for three parties not all effects are wirings. We have computed all the $\{0,1\}$-valued extremal effects for this case, finding $28886$ effects of which $2048$ are non-wirings.  These extremal effects break down into $66$ classes of which $3$ correspond to non-wirings. They are collected in the online Supplementary Material~\cite{supp}.  A representative of each of the classes of non-wiring are as follows:
\begin{align*}
  \innerp{e_1}{s}&=P(000|100)+P(001|010)+P(110|001)\\
  \innerp{e_2}{s}&=P(000|100)+P(001|010)+P(110|001)+P(010|000)\\
  \innerp{e_3}{s}&=P(000|100)+P(001|010)+P(110|001)+P(010|000)+P(101|000)\,.
\end{align*}

Performing a vertex enumeration using the software {\sc Porta}~\cite{porta}, we are also able to find all the extremal representations of the identity effect.  It turns out that there are 710760 of these, of which 744 are $\{0,1\}$-valued (and 680 of the latter are wiring representations).  These extremal representations of the identity effect break down into 307 classes, of which 9 are $\{0,1\}$-valued (and 8 of the latter are wiring representations).  Representatives of each of the 307 classes can be found in the Supplementary Material~\cite{supp}.

In principle these representations can be used to generate all extremal effects as suggested by Lemma~\ref{lem:areall}. However, this computation is prohibitively time consuming because many of the classes of extremal representations of the identity effect have $27$ non-zero entries, and for these there are $2^{27}$ candidate effects that can be formed by zeroing various entries (by contrast, the $\{0,1\}$-valued extremal representations of the identity effect have only $8$ non-zero entries, which is why we could solve this case). Nevertheless, to illustrate that this can in principle be done we have taken two of the identity effects with a smaller number of non-zero entries, namely $u^{\R}_1$ and $u^{\R}_2$ satisfying
\begin{align*}
\innerp{u^{\R}_1}{s}&=\frac{1}{3}\big(P(000|101) + 2P(000|110) + P(001|010) + P(001|011) + P(001|100) + P(010|010) + P(010|011) + \\
&\quad\qquad P(010|101) + P(011|100) + 2P(011|111) + P(100|100) + 2P(100|111) + P(101|010) + P(101|011) + \\&\quad\qquad
 P(101|101) + P(110|010) + P(110|011) + P(110|100) + P(111|101) + 2P(111|110)\big)\\
\innerp{u^{\R}_2}{s}&= \frac{1}{2}\big(P(000|011) + P(000|101) + P(001|011) + P(001|101) + P(010|101) + P(010|110) + P(011|010)+\\
&\quad\qquad P(011|101) + P(100|011) + P(100|100) + P(101|011) + P(101|110) + P(110|100) + P(110|111) + \\
&\quad\qquad P(111|010) + P(111|111)\big)\,,
\end{align*}
and computed all the sub effects from these (up to symmetry), using a linear program to remove any that are convex combinations of the extremal deterministic effects we already found. The output of this computation is given in the Supplementary Material~\cite{supp}. The following are two examples:
\begin{align*}
    \innerp{e_4}{s}&=\frac{1}{3}\left(P(001|010)+P(001|100)+P(010|011)+P(101|010)+P(111|101)+2P(111|110)\right)\\
    \innerp{e_5}{s}&=\frac{1}{2}\left(P(000|011)+P(001|011)+P(010|110)+P(011|010)+P(100|100)+P(111|010)\right)\,.
\end{align*}
    \label{sec:effect_examples}

\subsubsection{$\parties=4$}
In this case computing all the extremal representations of the identity effect is not feasible in reasonable time using {\sc Porta}. Furthermore the number of extremal effects is too large to directly use our previous technique for the $\{0,1\}$-valued extremal effects. Instead we can compute all the $\{0,1\}$-valued extremal effects by computing one representative of each symmetry class.  The computation works in the same way as before, but we remove symmetries at every step.

In particular, in Algorithm~\ref{alg:id} we add a step between Steps~\ref{st:3} and~\ref{st:4} that removes elements of $S$ that are equal to others under relabelling symmetries.  In Algorithm~\ref{alg:all}, rather than using the list of local values $L(e_i)$ we generate a canonical form of these by generating $L(e_i)$ for every symmetry of $e_i$ and then storing the first of all of these according to some ordering function (e.g., since each list $L(\Pi e_i)$ is $\{0,1\}$-valued, they can be ordered as a binary number). We run Algorithm~\ref{alg:all} with the modification to classify into wiring or non-wiring representations (i.e., using Steps~\ref{st4p} and~\ref{st5p}).

Overall we find $168301$ classes of extremal $\{0,1\}$-valued effect, of which $124698$ are wiring representations.  By generating all the symmetries of each, we can then compute the total number of $\{0,1\}$-valued effects to be $7940781474$, of which $4729832866$ are wiring representations. Because of the size, we only supply Supplementary files with an element of each class in this case~\cite{supp}.

\subsection{Generalisations: more inputs and outputs}

In the case $\ins=2$ and $\outs=2$ we were only able to partially solve the cases with $\parties=3$ and $\parties=4$.  Increasing the number of inputs and outputs further increases the complexity, but we can make a few remarks.

Firstly consider the case $\parties=2$. It was proven in~\cite{ShortBarrett} that the bipartite effect spaces also only contain wirings. Using our code we enumerate the number of classes for the first few cases, as well as the total number of effects (see Table~\ref{tab:numbers}).
\begin{table}
    \centering
    \begin{tabular}{c|ccc}
$\outs$&2&3&4\\\hline
classes $\ins=2$&7&44&523\\
classes $\ins=3$&7&48&-\\
classes $\ins=4$&7&-&-\\\hline
total $\ins=2$&82&8930&2977858\\
total $\ins=3$&248&43400&-\\
total $\ins=4$&562&-&-
\end{tabular}
    \caption{Numbers of classes and effects for two parties with various numbers of inputs and outputs. (We use - to indicate that we did not compute this case). The number of classes in the case $\outs=2$ will remain at 7 for any $\ins$.}
    \label{tab:numbers}
\end{table}
Data with the full set of extremal effects for these cases can be found in the Supplementary Material~\cite{supp}.

In the case $N=3$, the only additional case we attempt is $\ins=3$, $\outs=2$.  Here we use the previous method to compute the $\{0,1\}$-valued effects, obtaining $79$ classes of such effect of which $76$ are wirings and $3$ are non-wirings.  In total the number of $\{0,1\}$-valued effects is $505136$ which breaks down as $449840$ wirings and $55296$ non-wirings.  Again, these cases can be found in the Supplementary Material~\cite{supp}.

Our codes can also be used for the enumeration of all extremal effects in these scenarios (within the computational limitations). We omit  explicit characterisations here. 

\section{Applications of non-wirings} \label{sec:application}
In this section we discuss some applications of non-wiring measurements, focusing on those that outperform wirings.

\subsection{State discrimination}
Given a black box that outputs one of two possible (known) states, $s_1$ and $s_2$, with probability $1/2$ each, the task is to choose a measurement that gives the highest probability of correctly guessing which state was produced given just one copy.

When trying to discriminate between two probability distributions, $P_X$ and $Q_X$, the guessing probability is
$$\frac{1}{2}\left(1+D_{\mathrm{C}}(P_X,Q_X)\right),$$
where $D_{\mathrm{C}}$ is the total variation distance (the subscript C indicating classical), i.e., $D_{\mathrm{C}}(P,Q)=\frac{1}{2}\sum_x|P_X(x)-Q_X(x)|$. In the case of two quantum states, $\rho_1$ and $\rho_2$, this optimal guessing probability is
$$\frac{1}{2}\max_{E_1,E_2}\left(1+D_{\mathrm{C}}\left(\{\tr(E_1\rho_1),\tr(E_2\rho_1)\},\{\tr(E_1\rho_2),\tr(E_2\rho_2)\}\right)\right)=\frac{1}{2}\left(1+D_{\mathrm{Q}}(\rho_1,\rho_2)\right),$$
where $\{E_1,E_2\}$ form a POVM and $D_{\mathrm{Q}}$ is the trace distance (the subscript Q indicating quantum), i.e., $D_{\mathrm{Q}}(\rho_1,\rho_2)=\frac{1}{2}\tr\left|\rho_1-\rho_2\right|$ (see, e.g.,~\cite{Helstrom,Nielsen&Chuang}).

The analogous formula for boxworld is that the optimal probability is
$$\frac{1}{2}\max_{e_1,e_2}\left(1+D_{\mathrm{C}}\left(\{\innerp{e_1}{s_1},\innerp{e_2}{s_1}\},\{\innerp{e_1}{s_2},\innerp{e_2}{s_2}\}\right)\right)=\frac{1}{2}\max_{e_1}\left(1+\left|\innerp{e_1}{s_1}-\innerp{e_1}{s_2}\right|\right),$$
where the first maximization is over all measurements $\{e_1,e_2\}$ and the second is over all effects $e_1$. The optimum will always be achieved by an extremal effect, hence, in cases where we have computed all extremal effects, we can calculate it by running over all of these. By analogy with the quantum and classical cases, it is natural to define $D_{\mathrm{B}}(s_1,s_2):=\max_{e_1}|\innerp{e_1}{s_1}-\innerp{e_1}{s_2}|$ (it is not clear how to remove the maximization from this expression in this case). The quantity $D_{\mathrm{B}}$ also satisfies the requirements of a distance measure~\cite{ShortWehner}.

\bigskip

There are pairs of states that can be perfectly distinguished with non-wirings, but for which the same is not true when only wirings are considered. An example for three 2-input, 2-output systems are the states\footnote{See~\eqref{eq:state_rep} for the ordering of the components.}:
\begin{align}
s_1= (&177, 0, 183, 177, 183, 177, 0, 183, 88, 89, 0, 360, 184, 176, 92, 91,
184, 88, 176, 89, 92, 0, 91, 360, 52, 220, 36, 229,\nonumber\\
&54, 38, 222, 229, 0, 88, 92, 184, 360, 89, 91, 176, 52, 36, 54, 222, 220, 229, 38, 229, 54, 52, 38, 220, 222, 36, 229,\nonumber\\
&229, 34, 72, 72, 186, 72, 186, 186, 272)/1080 \label{eq:s1}\\
s_2= (&177, 0, 183, 177, 183, 177, 0, 183, 88, 89, 360, 0, 175, 185, 92, 91, 175, 88, 185, 89, 92, 360, 91, 0, 227, 36, 221, 53,\nonumber\\
&229, 223, 38, 53, 360, 88, 92, 175, 0, 89, 91, 185, 227, 221, 229, 38, 36, 53, 223, 53, 229, 227, 223, 36, 38, 221, 53,\nonumber\\
&53, 269, 187, 187, 72, 187, 72, 72, 34)/1080 \label{eq:s2}.
\end{align}
The optimal\footnote{Note that we cannot run over all non-wiring effects in this case because we have not computed them all. However, because we can find a non-wiring effect that distinguishes with probability $1$, this must be optimal.} non-wiring effect, which perfectly distinguishes these, is the deterministic effect
\begin{align*}
e_1=&(0,0,0,0,0,0,1,0,0,0,0,1,0,0,0,0,0,0,0,0,0,0,0,1,0,0,0,0,0,0,0,0,0,0,0,0,1,0,0,0,0,0,0,0,0,0,0,0,0,0,0,\phantom)\\
  &\phantom(0,0,0,0,0,0,0,0,0,0,0,0,0) 
  \end{align*}
(which satisfies $\innerp{e_1}{s}=P(110|000)+P(011|001)+P(111|010)+P(100|100)$ for a general state $s$) from which one can verify that $\innerp{e_1}{s_1}=1$ and $\innerp{e_1}{s_2}=0$.  By running over all wirings effects find that the maximum guessing probability for these two states using wirings is $5/6$.  This best possible distinguishing probability for wirings is achieved by the wiring
\begin{align*}
e'_1 = (&0, 0, 0, 0, 0, 0, 0, 0, 0, 0, 0, 0, 0, 0, 0, 0, 0, 0, 0, 0, 0, 0, 1, 
0, 1, 0, 1, 0, 1, 1, 0, 0, 0, 0, 0, 0, 0, 0, 0, 0, 0, 0, 0, 0, 0, 0, 
0, 0, 0, 0, 0, \\ &0, 0, 0, 0, 0, 0, 0, 0, 0, 0, 0, 0, 0).
\end{align*}

As the non-wiring extremal effects naturally split into those that are deterministic and those that are not, a natural question is whether the advantages that the deterministic non-wiring $e_1$ above exhibits for state discrimination can also be found for non-deterministic ones. We find that this is the case, in the sense that there are pairs of states that can be perfectly distinguished only by a non-deterministic non-wiring. Such an example are the states
\begin{align*}
t_1 = (&736, 753, 223, 1367, 1370, 205, 904, 922, 644, 845, 745, 845, 730, 845, 136, 1690, 455, 429, 504, 1691, 806, 1127,\\
&1468, 0, 884, 0, 505, 1690, 243, 1690, 623, 845, 416, 471, 1127, 804, 1690, 487, 0, 1485, 887, 0, 241, 1690, 487, 1690,\\
&640, 845, 416, 429, 1127, 846, 845, 1127, 845, 845, 0, 845, 1128, 845, 1127, 845, 0, 1690)/6480\\
t_2 = (&124, 124, 279, 0, 0, 275, 137, 141, 0, 248, 98, 181, 94, 181, 278, 0, 305, 124, 98, 0, 0, 0, 137, 416, 0, 429, 98, 0, 0, 0,\\
&372, 181, 124, 305, 0, 0, 0, 94, 416, 141, 0, 429, 0, 0, 94, 0, 376, 181, 124, 124, 0, 181, 181, 0, 235, 235, 0, 248, 0,\\
&181, 0, 181, 470, 0)/1080
\end{align*}
which are perfectly distinguished by the non-wiring defined by 
\begin{align*}
f_1= ( &0, 0, 0, 0, 0, 0, 0, 0, 0, 0, 0, 0, 0, 0, 0, 0, 0, 0, 0, 1, 0, 0, 0, 0, 1, 0, 0, 0, 1, 1, 0, 0, 0,
0, 0, 0, 1, 0, 0, 0, 1, 0, 1, \\
&1, 0, 0, 0, 0, 0, 0, 1, 0, 0, 1, 0, 0, 0, 0, 0, 0, 0, 0, 0, 1)/2.
\end{align*}
Using wirings \emph{and} deterministic non-wirings we can only correctly guess which of these two states is present with probability at most $2423/2592\approx0.935$.

We can turn this problem around and ask which non-wirings are advantageous for state discrimination (meaning that they outperform wiring effects for some pair of states). For any non-wiring effect $e$, this question can be answered, using a linear program. Let $s_1$ and $s_2$ be the two states to be distinguished and $\mu$ be fixed. The linear program is
\begin{align*}
{\min}_{s_1,s_2,\nu} \qquad  & \nu \\
\operatorname{subject \ to} \qquad 
& s_1, \ s_2\geq 0 \\
& v_{\operatorname{N}}.s_1 = 1, \ \quad \quad v_{\operatorname{N}}.s_2 = 1 \\
& M_{\operatorname{NS}}\,s_1 = 0, \quad \quad M_{\operatorname{NS}}\,s_2 = 0 \\
& e.(s_1-s_2) = \mu, \quad M_{\operatorname{W}}(s_1-s_2) \leq  \nu 
\end{align*}
where $v_{\operatorname{N}}$ is a vector that encodes the normalisation constraints, and $M_{\operatorname{NS}}$ and $M_{\operatorname{W}}$ are matrices encoding the non-signalling constraints\footnote{The rows of $M_{\operatorname{NS}}$ are the generators, $\{r_i\}$, of the no-signalling moves, for instance.} and wirings respectively. (An inequality between a vector and a number is interpreted element-wise.) A non-wiring $e$ is advantageous if $\nu<\mu$. In the case of perfect distinguishability $\mu$ is set to $1$. In order to find any separation, $\mu$ could be taken as a variable that is optimised and $\mu -\nu$ maximized instead.

We have implemented this program in {\sc Matlab}, relying on {\sc YALMIP}~\cite{yalmip} and {\sc MOSEK}~\cite{mosek} to solve the linear programs. Using this program, various examples analogous to the ones above can be found. Checking all effects mentioned in Section~\ref{sec:effect_examples} with this program, we find that there are also examples of non-wirings that only allow for distinguishing states perfectly that can also be perfectly distinguished with wirings.  In addition, many of the effects outperform wirings with respect to the distinguishing probabilities they achieve for some states, but without reaching perfect distinguishability.

\subsection{Nonlocality without entanglement}

The phenomenon of \emph{quantum nonlocality without entanglement}~\cite{BDFMRSSW} describes the situation in which a set of product effects are combined to make a measurement that cannot be implemented by local operations and classical communication. In a sense these measurements are the quantum analogue of non-wirings. Analogously to~\cite{BDFMRSSW}, we construct a set of eight states and an eight outcome measurement that can perfectly distinguish these, while they are not perfectly distinguishable with wirings.  Our states are: 
\begin{align*}
    s_1 &= t_3 \otimes t_4 \otimes t_2 \quad
    s_2 = t_1 \otimes t_4 \otimes t_2\quad
    s_3 = t_4 \otimes t_2 \otimes t_3\quad
    s_4 = t_4 \otimes t_2 \otimes t_1\\
    s_5 &= t_1 \otimes t_1 \otimes t_3\quad
    s_6 = t_1 \otimes t_3 \otimes t_1\quad
    s_7 = t_3 \otimes t_1 \otimes t_4\quad
    s_8 = t_3 \otimes t_1 \otimes t_2,
\end{align*}
where $t_1=\left(1,0\mid 1,0\right)$, $t_2=\left(1,0\mid 0,1\right)$, $t_3=\left(0,1\mid 1,0\right)$, $t_4=\left(0,1\mid 0,1\right)$ following the notation from \eqref{eq:notation}.
These can be perfectly distinguished by the measurement defined by the effects that for any state $s$ behave as
\begin{align*}
    \langle e_1,s \rangle  &= P(110|000) \quad
    \langle e_2,s \rangle  = P(011|001) \quad
    \langle e_3,s \rangle  = P(111|010)\quad
    \langle e_4,s \rangle  = P(100|100)\\
    \langle e_5,s \rangle  &= P(001|000)\quad
    \langle e_6,s \rangle  = P(010|001)\quad
    \langle e_7,s \rangle  = P(101|010)\quad
    \langle e_8,s \rangle  = P(000|100).
\end{align*}
To show that no wiring can achieve perfect distinguishability of these 8 states, we have enumerated all deterministic 8-outcome wiring measurements. We find that the maximum guessing probability achievable with a wiring when the states are chosen uniformly is $7/8$. This is achieved by the wiring with effects satisfying\footnote{This corresponds to: input $x_3=0$, if $a_3=0$, input $x_2=0$ then $x_1=a_2\oplus1$, while if $a_3=1$, input $x_1=0$ then $x_2=a_1$.}
\begin{align*}
    \langle e_1,s \rangle  &= P(110|000) \quad
    \langle e_2,s \rangle  = P(010|000) \quad
    \langle e_3,s \rangle  = P(111|010)\quad
    \langle e_4,s \rangle  = P(100|100)\\
    \langle e_5,s \rangle  &= P(001|000)\quad
    \langle e_6,s \rangle  = P(011|000)\quad
    \langle e_7,s \rangle  = P(101|010)\quad
    \langle e_8,s \rangle  = P(000|100).
\end{align*}

Our example disproves Observation~1 from~\cite{Banik}. Indeed we show that these 8 states cannot be perfectly distinguished with 1-way LOCC in boxworld, while a global measurement achieves this. We see this as an indication that it is not the local indistinguishability within pairs of local states that causes this phenomenon in this example but rather the existence of separable but global measurements.

Note that on the other hand an analogue to the bipartite (9-state) example from~\cite{BDFMRSSW, Groisman} cannot be constructed, since all bipartite measurements in box-world (even in higher dimensions) are wirings~\cite{ShortBarrett}. 

In boxworld all measurements are separable so a distinction between separable and entangled measurements cannot be made. That we can still demonstrate that separable measurements outperform wiring measurements suggests that in some contexts comparing separable and wiring measurements may be more natural than comparing separable and entangled measurements, although in others, e.g., when considering teleportation, entangled measurements are required.

\subsection{Nonlocality distillation}\label{sec:NLdist}
Consider two parties, Alice and Bob, who hold parts of $t$ bipartite systems, with each subsystem having two inputs and two outputs. For simplicity, take these $t$ systems to be identical (with state $\hat{s}$). A nonlocality distillation protocol seeks to use these $t$ systems to give a larger violation of a Bell inequality than is possible with only 1. The most general strategy is for each party to associate a $t$-system 2-outcome measurement with each possible input. Such measurements have the form $\{e,u-e\}$ and so can be expressed in terms of one effect. Thus, the overall strategy can be expressed using 4 effects that act on $t$ systems (one for each of Alice's inputs and one for each of Bob's inputs). Because the individual states are identical, the overall starting state is the $t$-fold tensor product, $s=\hat{s}^{\ot t}$. If $e_x$ are the effects associated with outcome $0$ when Alice's input is $x$, and $f_y$ are likewise those for Bob, then the outcome probabilities are given by $P'(00|xy)=\innerp{e_x\ot f_y}{s}$, $P'(01|xy)=\innerp{e_x\ot(u-f_y)}{s}$ etc., where the tensor factors need to be matched appropriately ($s\in\cS_{A_1B_1A_2B_2\ldots A_tB_t}$, $e_x\in\cE_{A_1A_2\ldots A_t}$ and $f_y\in\cE_{B_1B_2\ldots B_t}$). That the overall effect is a tensor product reflects the independence of Alice's and Bob's operations. The idea of nonlocality distillation is to choose the four effects $e_0$, $e_1$, $f_0$ and $f_1$ so as to maximize the violation of a Bell inequality in the resulting distribution, $P'(ab|xy)$. Figure~\ref{fig:oval} depicts this intuitively for $3$ shared systems.

\begin{figure}[h]
    \centering
    \includegraphics[scale=0.6]{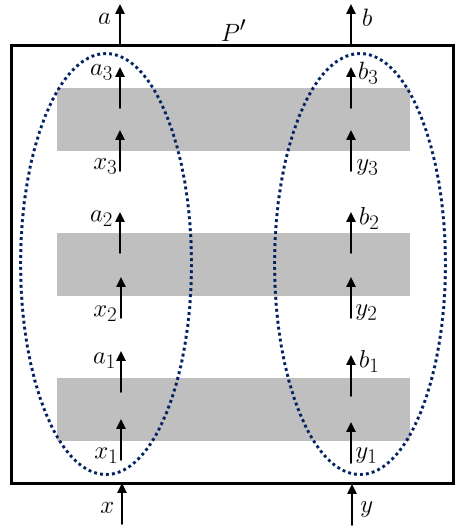}
    \caption{Illustration of nonlocality distillation in the case $t=3$. The grey boxes represent the identical initial bipartite states. The outer black frame represents the final correlations $P'$. The dashed ovals indicate the systems the measurements are performed on (left for Alice; right for Bob).}
    \label{fig:oval}
\end{figure}

For systems with two inputs and two outputs, the only extremal Bell inequality (up to symmetry) is the Clauser-Horne-Shimony-Holt (CHSH) inequality, which we can express as $\operatorname{CHSH}(P(ab|xy))=E_{00}+E_{01}+E_{10}-E_{11}$, with $E_{xy}=P(a=b|xy)-P(a\neq b|xy)$. We hence use this as our measure of nonlocality.
Since optimizing over all effects for one party is a linear program, we can run over all extremal effects for Alice and do a linear program for Bob to determine the optimal strategy for CHSH-value distillation given $3$ shared systems~\cite{eftaxias2022advantages} (in principle we could also do $4$, but, given the number of effects, the computation time is prohibitive). 

We perform the optimization for states taken from two 2-dimensional cross sections of the set of non signalling distributions. To describe these cross sections, we use 
$P^{\rm L}_i(ab|xy)=\delta_{a, \mu x \oplus \nu}\, \delta_{b, \sigma y \oplus \tau}$ for $\mu, \nu,\sigma,\tau \in \{0,1\}$, $i=1+\tau+2\sigma+4\nu+8\mu$ to enumerate the set of $16$ local deterministic boxes and $P^{\rm NL}_i(ab|xy)=\frac{1}{2}\delta_{a \oplus b, x y \oplus \mu x \oplus \nu y \oplus \sigma}$ for $\mu, \nu, \sigma \in \{0,1\}$, $i=1+\sigma+2\nu+4\mu$ to denote the set of $8$ extremal no-signalling boxes~\cite{Cirelson93,PR}.  In terms of these, the cross sections we have examined are:
\begin{align}
\mathrm{CS~I:}&\ \omega P^{\mathrm{NL}}_1+\frac{\eta}{2}( P^{\mathrm{L}}_1+P^{\mathrm{L}}_6)+(1-\omega-\eta)P^{\mathrm{O}} \nonumber \\
\mathrm{CS~III:}&\ \omega P^{\mathrm{NL}}_1+\frac{\eta}{2} (P^{\mathrm{L}}_1+P^{\mathrm{L}}_9)+(1-\omega-\eta)P^{\mathrm{O}}\,,
\label{CSequations} \end{align}
where $P^{\mathrm{O}}=3/4P^{\mathrm{NL}}_1+1/4P^{\mathrm{NL}}_2$ is local and $\eta,\omega\geq0$ with $\eta+\omega\leq1$. The labelling of these cross sections is chosen to follow~\cite{eftaxias2022advantages}, and the reason we use them is because they have been considered previously.

For the case $t=3$, the effects found to give rise to CHSH distillation are given below (expressed in terms of their inner product with an arbitrary tripartite state, represented by the ovals in Figure~\ref{fig:oval}). On Alice's side
\begin{align}
  \innerp{e_0}{s}&=P(000|000)+P(011|000)+P(101|000)+P(110|000) \nonumber \\
  \innerp{u-e_0}{s}&=P(001|000)+P(010|000)+P(100|000)+P(111|000) \nonumber \\
  \innerp{e_1}{s}&=P(111|011)+P(100|101)+P(001|110)+P(010|111) \nonumber \\
  \innerp{u-e_1}{s}&=P(011|011)+P(110|101)+P(000|110)+P(101|111) \label{effectAlice}
\end{align}
while on Bob's
\begin{align}
  \innerp{f_0}{s}&=P(011|011)+P(110|101)+P(000|110)+P(101|111) \nonumber \\
  \innerp{u-f_0}{s}&=P(111|011)+P(100|101)+P(001|110)+P(010|111) \nonumber \\
  \innerp{f_1}{s}&=P(101|000)+P(000|001)+P(110|010)+P(011|100) \nonumber \\
  \innerp{u-f_1}{s}&=P(010|000)+P(001|001)+P(100|010)+P(111|100). \label{effectBob}
\end{align}
Note that $e_1$, $f_0$ and $f_1$ are not wirings.
Referring again to Figure~\ref{fig:oval}, three systems characterised by identical states, each corresponding to coordinates $(\eta,\omega)$ in CS~I or CS~III and initial nonlocality $\text{CHSH}_{\text{i}}=2(1+\omega)$, get distilled through the protocol consisting of the effects \eqref{effectAlice} and \eqref{effectBob}, and the final nonlocality (that of the $P'$ correlations) is
\begin{align}
    \text{CHSH}_{\text{f}}^{\text{I}}&= \frac{1}{32}\Big[7\omega^3-15\eta^3+33\omega^2+57\omega+3\eta^2(7+11\omega)+3\eta(9+26\omega+13\omega^2)+31\Big] \label{post_chsh_i} \\[8pt]
    \text{CHSH}_{\text{f}}^{\text{III}}&= \frac{1}{32}\Big[7\omega^3+5\eta^3+33\omega^2+57\omega+\eta^2(13+25\omega)+3\eta(5+18\omega+9\omega^2)+31\Big]. \label{post_chsh_iii}
\end{align}
The distillable region (where $\text{CHSH}_{\text{f}}> \text{CHSH}_{\text{i}}$) for the protocol is shown in Figure~\ref{distill_i_iii} for the two cross sections, CS~I and CS~III. 

\begin{figure}[h]
\centering
\begin{minipage}[t]{0.47\columnwidth}
	\includegraphics[scale=0.55]{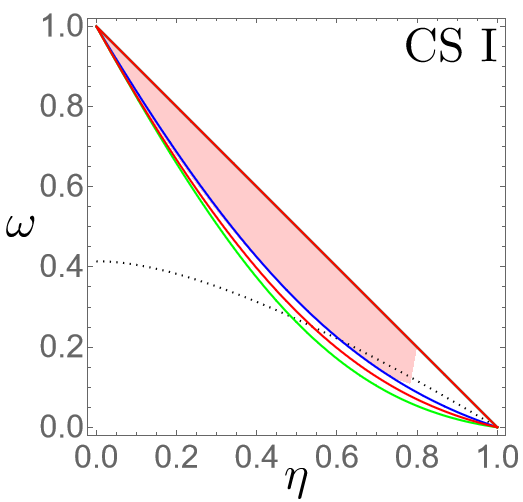}
\end{minipage}
\begin{minipage}[t]{0.05\columnwidth}
	
\end{minipage}
\begin{minipage}[t]{0.47\columnwidth}
	\includegraphics[scale=0.55]{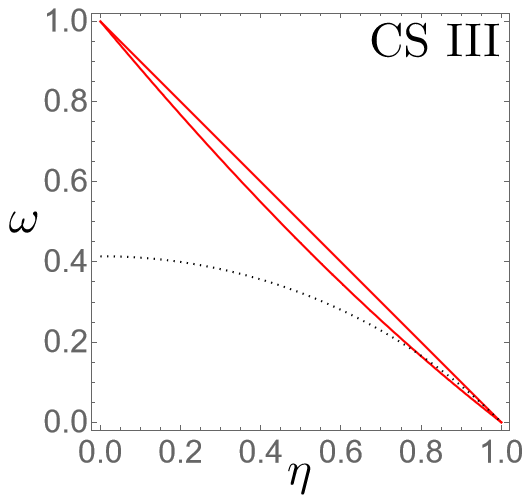}
\end{minipage}
\caption{The red lines show the boundary of the distillable regions using our non-wiring protocol, in different cross-sections of the (2,2,2) no-signalling polytope. The dotted curves represent the boundary of the quantum realizable correlations. On the left, the blue curve represents the boundary of the distillable region using the {\sc xor} protocol of~\cite{Allcock}, while the green one represents correlations that can be distilled through the {\sc or} protocols of~\cite{eftaxias2022advantages}. The red shaded area shows where our non-wiring protocol achieves higher final CHSH values than the protocols of~\cite{eftaxias2022advantages,Allcock,Hoyer,Forster2009} (for the protocol of~\cite{Forster2009}, the comparison has been made with both its 2-copy and 3-copy versions). Note that in CS~III there is no two-copy protocol that can distil any of the states~\cite{eftaxias2022advantages}.}
\label{distill_i_iii}
\end{figure}

\subsection{Limitations for information processing with wirings and boxworld's non-wiring operations} 

Despite the advantages we managed to demonstrate above, access to non-wirings in boxworld does not unlock the same potential as access to measurements that cannot be implemented as local measurements and classical communication in quantum mechanics does. In particular, all measurements in boxworld are separable~\cite{ShortBarrett} (i.e., can be expressed as a sum of product effects), which leads to various restrictions. For instance, teleportation and entanglement swapping are impossible in boxworld~\cite{barrett,ShortBarrett,SPG}. This directly implies that entanglement swapping is also not possible in the multi-partite setting (i.e., with non-wirings). A direct proof of this, which also applies to other GPTs, is obtained by following the same lines of reasoning as Lemma~2 of~\cite{WC5}.

\subsection{Boxworld effects in other GPTs} \label{sec:other_gpts}

Wirings correspond to a classical processing of inputs and outcomes from measurements and as such the inputs and outcomes of \emph{any} GPT can be connected by wirings. We show here that from the effects derived in this work ---  including non-wiring effects --- we can indeed construct valid effects for any other GPT.

Consider a single system that is fully characterised by $\ins$ inputs and $\outs$ outcomes and call the SBV effects $e_{1,1},\ldots,e_{\ins,\outs}$, where $e_{i,j}$ is the effect that has an entry $1$ for the $i^{\text{th}}$ input and $j^{\text{th}}$ outcome and is zero otherwise. For $\parties$ parties, the effects $\{e_{i_1,j_1} \otimes \cdots \otimes e_{i_\parties,j_\parties}\}_{i_1,\ldots, i_\parties, j_1, \ldots, j_\parties}$ span the effect cone of all boxworld effects. Thus, the effects we derived above for $\parties$-parties can be written in terms of the local effects as 
$$e=\sum_{{i_1},\ldots, i_\parties =1}^{\ins} \sum_{{j_1},\ldots, j_\parties =1
}^{\outs} \lambda_{{i_1}, \ldots, {i_\parties},j_1, \ldots, j_\parties} e_{i_1,j_1} \otimes \cdots \otimes e_{i_\parties,j_\parties},$$
where $\lambda_{{i_1}, \ldots, {i_\parties},j_1, \ldots, j_\parties}$ are some coefficients (which can be taken to be positive cf.\ Lemma~\ref{lem:pos}).
Now, as we know that any valid effect is valid on any valid state, i.e., $0\leq\innerp{e}{s}\leq1$, we also have 
$$0\leq \sum_{{i_1},\ldots,i_\parties=1}^{\ins} \sum_{{j_1},\ldots, j_\parties=1}^{\outs} \lambda_{{i_1}, \ldots, {i_\parties},j_1, \ldots, j_\parties} P(j_1, \ldots, j_\parties|i_1, \ldots, i_\parties) \leq 1,$$
where $P(j_1, \ldots, j_\parties|i_1, \ldots, i_\parties)$ is \emph{any} no-signalling distribution. This holds because the state space in box-world is in 1:1 correspondence with the set of all no-signalling distributions of which the effects  $\{ e_{i_1,j_1} \otimes \cdots \otimes e_{i_\parties,j_\parties} \}_{i_1,\ldots, i_\parties, j_1, \ldots, j_\parties}$ essentially just pick out elements.

Performing local measurements on $\parties$ party states leads to non-signalling correlations in any GPT, since non-signalling is one of the underlying assumptions. This implies that the correlations arising from performing any local measurements on an $\parties$ party system, are mapped to a probability by the map defined by the coefficients $\{\lambda_{{i_1}, \ldots, {i_\parties},j_1, \ldots, j_\parties}\}_{i_1,\ldots, i_\parties, j_1, \ldots, j_\parties}$. 
This means that for any GPT we can build valid effects from the ones we derived for boxworld, namely using these same coefficients:
$$f=\sum_{{i_1},\ldots, i_\parties=1}^{\ins} \sum_{{j_1},\ldots, j_\parties =1
}^{\outs} \lambda_{{i_1}, \ldots, {i_\parties},j_1, \ldots, j_\parties} f_{i_1,j_1} \otimes \cdots \otimes f_{i_\parties,j_\parties},$$
where the $f_{i,j}$ are local effects in the GPT of interest.

Note that a special case of the above is that in any GPT in which SBVs are valid local effects, all the effects found by our algorithms are directly valid.  Furthermore, the state space of a GPT can always be expressed in terms of the outcome probabilities of a set of fiducial measurements. Having made the transformation needed for this representation, the SBVs are local effects.

\section{The significance of wiring and non-wiring operations beyond GPTs}
\label{sec:beyond_GPT}
The polytope spanned by all identity effects in boxworld in any $(\parties,\ins,\outs)$ scenario is of interest beyond the scope of boxworld and even beyond the scope of GPTs. This is illustrated with the following two connections to other areas of quantum foundations. This means that our Algorithm~\ref{alg:all} may be of more general interest in the future in the sense that it allows us to construct logically consistent classical processes without causal order as well as compositions of contextuality scenarios.

\subsection{Logically consistent classical processes} \label{sec:class_without_causal_order}

We also remark that the polytope of all identity effects in boxworld is equivalent to that of the logically consistent classical processes, which was computed in~\cite{BaumelerWolf}.\footnote{This includes not only the $\{0,1\}$-valued identity effects but the whole polytope.} Since extremal effects for a single system in boxworld are SBVs in the $(\ins\outs)^\parties$-dimensional representation (and coarse-grainings of these), operationally, any extremal local measurement can be performed by choosing an input and then potentially coarse-graining the outcome, which is a classical process.\footnote{Note that in the language of~\cite{BaumelerWolf} a process being classical does not necessarily mean that the theory considered is classical in the GPT sense, i.e., classical probability theory.} This explains the correspondence between performing multi-system measurements on boxworld systems and the ways these classical local operations can be consistently connected. This correspondence extends: for any number of inputs and outputs the logically consistent classical processes can be understood in terms of the (identity) measurements of a GPT, namely boxworld. This  operational way of thinking about the logically consistent operations may aid our understanding of such operations.

\subsection{Composing contextuality scenarios}

Contextuality is a notion that captures the lack of predetermined outcomes in quantum measurements, as these may depend on the context in which the measurements are performed. There are various approaches capturing this notion, starting from the original work of Kochen and Specker~\cite{KS}. 
In~\cite{Wolfe_Sainz}, the composition of contextuality scenarios was analysed, following the approach to contextuality from~\cite{Acin_Contextuality}. According to this approach, a contextuality scenario is represented by a hypergraph $(\verts,\edges)$, where the vertices $\vertx\in\verts$ represent events of a specific outcome occurring and hyperedges $\edge\in\edges$ correspond to collections of events that make up a measurement.\footnote{In graph theory, vertices are usually $V$ and edges $E$; the notation is different here because in a noncontextual model vertices correspond to effects and hyperedges to measurements.} These edges may overlap on some vertices, meaning that the respective event can occur as part of several measurements. In a non-contextual model, each vertex can be assigned a probability, since the event occurs with the same frequency no matter which of the measurements it is part of is performed (in a contextual model the probability will in general also depend on the measurement). The probabilistic models that are allowed then further depend on the underlying theory (e.g.\ classical, quantum or more general).

These contextuality scenarios can then also be considered in the multi-party regime, meaning that independent contextuality scenarios for several independent systems are turned into a single scenario for the joint system. This means that for two systems $A$, $B$ for each of which a contextualtiy scenario $(\verts_A,\edges_A)$, $(\verts_B, \edges_B)$ is given, a joint scenario $(\verts_{AB},\edges_{AB})$ with events $\vertx_{AB}=(\vertx_A,\vertx_B) \in\verts_{AB}$, $\vertx_A\in\verts_A$, $\vertx_B\in\verts_B$ is constructed in a way that the probabilistic models defined on the hypergraph satisfy the non-signalling principle. Depending on the set $\edges_{AB}$ that is constructed for this purpose, we speak about a different product. 
While in the case of two contextuality scenarios there is a unique way to compose them~\cite{Foulis_Randall}, namely by means of the Foulis-Randall (FR) product, in~\cite{Acin_Contextuality,Wolfe_Sainz} different ways to compose more than two such scenarios, all respecting the non-signalling principle, were proposed. 

Our multi-partite $\{0,1\}$-valued effects can be seen as ways to combine products of deterministic local effects into multipartite measurements for our (non-signalling) GPT systems, which indeed amounts to the same mathematical problem as composing contextuality scenarios. Local deterministic effects correspond to the events $v$ and the identity effects to the edges $e$ in these multiparty scenarios. 
Thus our Algorithm~\ref{alg:id} can be seen as a way to construct compositions of contextuality scenarios. Specifically, this algorithm constructs a product known as the \emph{disjunctive FR product} in~\cite{Wolfe_Sainz}. The subset of wiring effects identified by Algorithm~\ref{alg:wir} constructs the \emph{maximal FR product} from~\cite{Wolfe_Sainz}.

The existence of extremal identity effects that are not $\{0,1\}$-valued further show that there is a more general way to combine single systems effects compatibly with non-signalling. This reasoning could also be applied to the events in a contextuality scenario. This suggests that the hypergraph formalism for describing contextuality scenarios needs to be extended to include a new concept that generalises the notion of a hyperedge. This generalisation requires a weight for each element of the hyperedge. [This can be easily added within the matrix representation of the hypergraph which works as follows. Each column corresponds to a vertex and each row denotes a hyperedge with a 1 meaning the vertex of the corresponding column is in the hyperedge and a 0 meaning it is not. In the generalisation, the matrix is no longer restricted to contain only elements of $\{0,1\}$.]  It would be interesting to explore the significance of this for non-contextuality in more detail.

\section{Conclusion} \label{sec:conclusion}
Characterizing measurements in theories beyond quantum theory allows us to better explore the possibilities for information processing in such theories and, in turn helps us understand what is special about quantum theory itself.  In this paper we have explored ways to generate all the effects present in a maximally nonlocal alternative theory, boxworld.  We have been able to find all the deterministic effects in several scenarios, dividing them into wirings and non-wirings, and also found many classes of non-deterministic effect. Although we have focused on boxworld, theories with fewer states have fewer constraints on allowed effects (under the no-restriction hypothesis) and hence the effects we have found are applicable in a wide range of GPTs (see Section~\ref{sec:other_gpts}).\footnote{Removing states can also lead to non-separable effects~\cite{Skrzypczyk_2009}.}

The effects we have found are also relevant for studies of logically consistent classical processes and for compositions of contextuality scenarios, where our findings suggest the need to extend the hypergraph formalism for dealing with these.

We have applied our findings to several applications, demonstrating advantages of both deterministic and non-deterministic non-wirings. In particular, we showed that, contrary to previous claims, examples of nonlocality without entanglement also appear in boxworld. We remark here that the \emph{existence} of further examples of quantum non-locality without entanglement has recently been shown in the literature, using a construction based on classical processes without causal order~\cite{Ravi,Murao}. Due to the correspondences we establish in Sections~\ref{sec:class_without_causal_order} and~\ref{sec:other_gpts},
any such example can likely be turned directly into an example for other GPTs.


In quantum theory there is often focus on entangled measurements, but there are other types of joint measurement (see, e.g.,~\cite{BDFMRSSW}) whose power would be useful to explore. Since boxworld has no non-separable effects and is somewhat limited with respect to its reversible dynamics~\cite{GMCD}, it would be interesting to complement our investigations with those of possible measurements in no-signalling theories with restricted nonlocality.

\acknowledgements
MW would like to thank \"Amin Baumeler for discussions about logically consistent classical processes and Elie Wolfe for pointing out his related work~\cite{Wolfe_Sainz}. GE is supported by the EPSRC grant EP/LO15730/1. MW is supported by the Lise Meitner Fellowship of the Austrian Academy of Sciences (project number M 3109-N). Part of this work appeared in GE's thesis~\cite[Chapter 7]{GiorgosThesis}.


\end{document}